\documentclass{article}

\usepackage{fullpage}
\usepackage[T1]{fontenc}
\usepackage{lmodern}
\usepackage{graphicx}
\usepackage{algorithm2e}
\usepackage{amssymb}
\usepackage{amsthm}
\usepackage{amsmath}
\usepackage{wrapfig}
\usepackage{lipsum}
\usepackage{footnote}
\usepackage{hyperref}
\makesavenoteenv{tabular}
\makesavenoteenv{table}

\newtheorem{proposition}{Proposition}[section]
\newtheorem{theorem}[proposition]{Theorem}
\newtheorem{lemma}[proposition]{Lemma}
\newtheorem{corollary}[proposition]{Corollary}

\newtheorem{definition}[proposition]{Definition}
\newcommand{\E}{\mathrm{E}}

\newcommand{\OPT}{{\sf OPT}}

\def\lc{\left\lceil}   
\def\rc{\right\rceil}

\usepackage[utf8]{inputenc}
\usepackage[T1]{fontenc}
\usepackage{hyperref}
\usepackage{url}
\usepackage{booktabs}
\usepackage{amsfonts}
\usepackage{nicefrac}
\usepackage{microtype}

\title{Differentially Private Clustering via Maximum Coverage}

\author{
	Matthew Jones \qquad   Huy L\^{e} Nguy\~{\^{e}}n \qquad Thy Nguyen \\
	Khoury College\\
	Northeastern University\\
	Boston, MA 02115 \\
	\texttt{\{jones.m,hu.nguyen,nguyen.thy2\}@northeastern.edu} \\
}

\begin{document}

\RestyleAlgo{boxruled}
\LinesNumbered

\maketitle

\begin{abstract}
    This paper studies the problem of clustering in metric spaces while preserving the privacy of individual data. Specifically, we examine differentially private variants of the k-medians and Euclidean k-means problems. We  present polynomial algorithms with constant multiplicative error and lower additive error than the previous state-of-the-art for each problem. Additionally, our algorithms use a clustering algorithm without differential privacy as a black-box. This allows  practitioners to control the trade-off between runtime and approximation factor by choosing a suitable clustering algorithm to use.
\end{abstract}

\section{Introduction}

In this work, we study the problem of clustering while preserving the privacy of individuals in the dataset.  Clustering is an important routine in many machine learning tasks, such as image segmentation \cite{yang2017towards, patel2013latent}, collaborative filtering \cite{mcsherry2009differentially, schafer2007collaborative}, and time series analysis \cite{mueen2012clustering}. Thus, improving the performance of private clustering has a great potential for improving other private machine learning tasks. We consider the clustering problem with differential privacy, which is a privacy framework that requires the algorithm to be insensitive to small changes in the dataset \cite{dwork2014algorithmic}. Formally, suppose we are given a metric $d$, a set $V$ of $n$ points in the metric, and a (private) set of demand points $D \subseteq V$.  Our objective is to choose a set $F \subset V$ of size $k$ to minimize the following objective: 
\begin{equation}
    \sum_{v \in V} \min_{f \in F} d(v, f)^p
\end{equation}
This is the $k$-medians problem when $p=1$ and the $k$-means problem when $p=2$. 

The $k$-medians problem has been studied extensively in the literature. Previous work in \cite{kariv1979algorithmic} proved that $k$-medians is NP-hard. There is a long line of works on approximation algorithms for $k$-medians without differential privacy \cite{chrobak2006reverse,arya2004local,charikar2002constant, jain2001approximation, jain2003greedy,BPRST17}. The state of the art is a $2.675+\epsilon$ approximation by~\cite{BPRST17}. The private $k$-medians problem was studied in \cite{gupta2010differentially}, which shows that any $(\epsilon_p,0)$-differentially private algorithm for the $k$-median problem must have cost at least $\OPT + \Omega (\Delta \cdot k \ln (n \slash k) \slash \epsilon_p)$, where $\OPT$ is the optimal cost. The paper also provided a polynomial time  $(\epsilon_p,0)$-differentially private algorithm that solves the problem with with costs at most $6\OPT  + O(k^2 \log^2 n\slash\epsilon_p)$. \cite{stemmer2018differentially} states a variant of the algorithm for $(\epsilon_p, \delta_p)$-differential privacy with cost
\begin{equation*}
  O \Big(  \OPT + \frac{\Delta k^{1.5} }{\epsilon_p} \log\frac{n}{\beta} \sqrt{\log n \cdot \log(1\slash \delta_p)}\Big)
\end{equation*} 
where $\beta$ is the failure probability.

\textbf{Our contribution} is a new polynomial time $(\epsilon_p, \delta_p)$-differentially private $k$-medians algorithm with constant multiplicative factor and improved additive error:
\begin{equation*}
     O\left(\OPT + \frac{k\Delta}{\epsilon_p}  \log n \log\left(\frac{e}{\delta_p}\right)\right)
\end{equation*}

Note that our additive error is linear in $k$ instead of $k^{1.5}$ as in the previous work by \cite{gupta2010differentially,stemmer2018differentially}, and almost matches the lower bound of \cite{gupta2010differentially} up to lower order terms (their original lower bound is for $(\epsilon_p,0)$-privacy but it can be extended to $(\epsilon_p,\delta_p)$ privacy with inverse polynomial $\delta_p$). By using the non-private algorithm of \cite{BPRST17} as a subroutine, our multiplicative factor is $6.35+\epsilon$, which is slightly worse than 6 from~\cite{gupta2010differentially}.

The $k$-means problem has also been studied extensively with a long line of works on privacy preserving approximation algorithms~\cite{BDMN05,NRS07,FFKN09,feldman2017coresets,balcan2017differentially,stemmer2018differentially}.
As an extension of our techniques, we also provide an $(\epsilon_p, \delta_p)$-differentially private algorithm for the Euclidean $k$-means problem with constant multiplicative error and better additive cost compared to previous work (see table \ref{tab:compare}). Our algorithm has cost:
\begin{equation*}
O\big(\OPT  +  k\log n +d^{0.51} (\log\log n)^{2.53} \left(k\log k \right)^{1.01} \big)
\end{equation*}

Again in this setting, our additive error is almost linear in $k$ as opposed to $k^{1.5}$ in previous work~\cite{stemmer2018differentially}.

In addition to improved performance guarantee, our $k$-medians and Euclidean $k$-means algorithms use a non-private clustering algorithm as a black-box. This allows practitioner to control the trade-off between approximation guarantee and runtime of the clustering algorithm and even use heuristics with good empirical performance. 

\textbf{Our techniques} for $k$-medians and Euclidean $k$-means include two main steps. In the first step, we iterate through distance thresholds from small to large and apply a differentially private Maximum Coverage algorithm to select centers that cover almost as many points as the optimal solution at those thresholds. For the second step, we create a new dataset based on the potential centers and apply a non-private clustering algorithm on this dataset. The new dataset is created by moving each demand points to the nearest potential center, and then applying the Laplace mechanism \cite{laplacenoise} to report the number of points at each potential center. This makes sure that privacy is preserved for the new dataset and no additional privacy cost when we apply the non-private clustering algorithm in the final step.

\section{Related Works}

\begin{table}[]
\begin{tabular}{l|c|c|c}
\textbf{Reference}           & \textbf{Objective}  & \textbf{Multiplicative Error} & \textbf{Additive Error}\\ \hline
Gupta et al. \cite{gupta2010differentially},&    $k$-medians\footnote[1]{The bound in \cite{gupta2010differentially} is for  $\epsilon$-differential privacy, while \cite{stemmer2018differentially} is for $\delta$-approximate $\epsilon$-differential privacy.}              & O(1) & $O(k\log n)^2$        \\
 Stemmer et al.  \cite{stemmer2018differentially} &&&  $O(k \log n)^{1.5}$ \\ \hline
\textbf{Ours} & $k$-medians  & $O(1)$ & $O(k\log n)$ \\ \hline  
Feldman et al. \cite{feldman2017coresets}  & Euclidean $k$-means & $O(k \log n)$ & $O\left(k\sqrt{d} \log (nd) \cdot 9 ^{\log^* \left(|X| \sqrt{d}\right)}\right) $ \footnote[2]{ \cite{feldman2017coresets} studies the discrete $d$-dimensional space $X^d$.}\\ \hline
Balcan et al. \cite{balcan2017differentially}  & Euclidean $k$-means & $O\left(\log^3n\right)$ & $O\left(\left(k^2+d\right)\log^5 n \right)$ \\ \hline
Stemmer et al. \cite{stemmer2018differentially} &  Euclidean $k$-means & $O(1)$ & $ O\big((k \log \left(n\log k\right))^{1.5}$  \\ &&& $+d^{0.51} (\log\log n)^{2.53} \left(k\log k \right)^{1.01} \big)$
\\ \hline 
 \textbf{Ours}& Euclidean $k$-means& $O(1)$ & $O\big(k\log n $ \\ 
 &&& $+d^{0.51} (\log\log n)^{2.53} \left(k\log k \right)^{1.01} \big)$\\
\end{tabular}
\caption{\label{tab:compare}
Comparison of our clustering algorithms with prior works, omitting dependence on $\epsilon_p, \delta_p$.}
\end{table}
Table \ref{tab:compare} summarizes the performance of previous work in comparison with our algorithms. In the table, only the work by \cite{gupta2010differentially} is for $(\epsilon_p,0)$-differential privacy, while the others are for $(\epsilon_p, \delta_p)$-differential privacy.  As mentioned above, \cite{gupta2010differentially} gave the first private $k$-medians algorithm with a constant multiplicative approximation and additive error polynomial in the number of centers and logarithmic in the number of points. The algorithm uses the local search approach of \cite{arya2004local}. Our algorithm, on the other hand, can be used with any non-private $k$-medians algorithm. 

For the Euclidean $k$-means problem, \cite{balcan2017differentially} proposes the strategy of first identifying a set of potential centers with low $k$-means cost, then applying the techniques of \cite{gupta2010differentially} to find the final centers among the potential centers. However, their  potential centers are only guaranteed to contain a solution with multiplicative approximation $O\left(\log^3 n\right)$. This result was improved by \cite{stemmer2018differentially}, which can construct a set of potential centers containing a solution with constant multiplicative approximation.

Another approach for the $k$-means problem is via the $1$-cluster problem. Given a set of input points in $\mathbb{R}^d$ and $t\leq n$, the goal is to find a center that covers at least $t$ points with the smallest radius. The work of \cite{feldman2017coresets} shows that the $k$-means problem can solved by running the algorithm for the $1$-cluster problem multiple time to find several balls to cover most of data points with $O(k\log n)$ multiplicative error. \cite{nissim2018clustering} proposed an improved algorithm for the $1$-cluster problem, resulting in a differentially private $k$-means algorithm with $O(k)$ multiplicative error.

\section{Preliminaries}

\subsection{Differential privacy}
Differential privacy is a privacy definition for computations run
against sensitive input data sets. Its requirement, informally, is that the computation behaves similarly on two input dataset that are nearly identical. Formally, 

\begin{definition} (\cite{dwork2006our})
A randomized algorithm $M$ has $\delta_p$-approximate $\epsilon_p$-differential privacy, or $(\epsilon_p,\delta_p)$-differential privacy, if for any two input sets $A$ and $B$ with a symmetric difference which has a single element, and for any set of outcomes $S \subseteq Range(M)$,
    \begin{align*}
\Pr[M(A) \in S]  \le  \exp(\epsilon_p) \times \Pr[M(B) \in S] + \delta_p\; .
  \end{align*}
\end{definition}

If $\delta_p = 0$, we say that $M$ is $\epsilon_p$-differentially private.
An algorithm with $(\epsilon_p,0)$-differntial privacy ensures that the output $M(A)$ is (almost) equally likely to be observed on neighboring datasets, whereas in $(\epsilon_p,\delta_p)$-differential privacy, $\delta_p$ dictates the probability that $\epsilon_p$-privacy fails to hold \cite{dwork2006our}. In this way, $(\epsilon_p,\delta_p)$-differential privacy is a relaxation of $\epsilon_p$-differential privacy. We use an error parameter $\epsilon$ for utility and we use $\epsilon_p$ and $\delta_p$ to denote parameters for differential privacy (or $\epsilon_s$ and $\delta_s$ for algorithm \ref{alg:setcover}, to differentiate privacy parameters to different algorithms).

One of the most basic constructions for differentially private algorithms is the Laplace mechanism. 

\begin{definition} ($L_1$ sensitivity)
A function $f: \mathbb{N}^{|X|} \rightarrow \mathbb{R}^k  $ has $L_1$ sensitivty $\Delta f$ if 
$ \| f(A) - f(A') \|_1 \leq \Delta f$ for all $A,A'$ with a symmetric difference which has a single element.
\end{definition}

\begin{theorem} (Laplace mechanism  \cite{laplacenoise})
\label{thm:laplacemech}
Let function $f: \mathbb{N}^{|X|} \rightarrow \mathbb{R}^k  $ have $L_1$ sensitivty $\Delta f$ and $\epsilon_p > 0$. Mechanism $M$ that on input $A$ outputs  $f(A) + Lap(\frac{\Delta f}{\epsilon_p})$ is $(\epsilon_p,0)$-differentially private, where $Lap(\frac{\Delta f}{\epsilon_p})$ denotes a random variable following Laplace distribution with scale parameter $b = \frac{\Delta f}{\epsilon_p}$.
\end{theorem}

Another tool for construction of differential private algorithm we use in this work is the exponential mechanism. This construction is parameterized by a query function $q(A,r)$ mapping a pair of input data set $A$ and candidate result $r$ to a real valued. With $q$ and a privacy value $\epsilon_p$, the mechanism selects an output in favor of high score value:
  \begin{align}
    \label{eqn:expmech}
Pr[{\mathcal E}_q^\epsilon(A) = r]  \propto & \exp(\epsilon_p q(A,r)) 
  \end{align}
\begin{theorem}
\label{theorem:expLoss} (\cite{mcsherry2007mechanism})
  The exponential mechanism, when used to select an output $r \in R$ gives $2\epsilon_p\Delta$-differential privacy, letting $R_\OPT$ be the subset of $R$ achieving $q(A,r) = \max_r q(A,r)$, ensures that
\begin{align*}
    \Pr[q(A, \mathcal{E}_q^\epsilon(A)) < \max_r q(A,r) -
    \ln(|R|/|R_\OPT|)/\epsilon_p - t/\epsilon_p] \leq \exp(-t)
 \end{align*}

\end{theorem}

\subsection{Maximum Coverage}

 \begin{wrapfigure}[12]{R}{0.5\textwidth}

\begin{algorithm}[H]
\caption{Maximum Coverage}
\label{alg:setcover}
\DontPrintSemicolon
\LinesNotNumbered
\SetAlgoLined
\KwIn{Set system $(U,\mathcal{S})$, a private set $R \subset U$ to cover, $\epsilon_s, \delta_s$, $m$ }

 $i \leftarrow 1$, $R_i = R$, $\mathcal{S}_i
    \leftarrow \mathcal{S}$. $\epsilon' \leftarrow \epsilon_s/2\ln(\frac{e}{\delta_s})$.

    \For{$i=1,2,\ldots,m$} {
  Pick a set $S$ from $\mathcal{S}_i$ with probability
    proportional to $\exp(\epsilon' |S \cap R_i|)$.
    
Output set $S$.
 
 $R_{i+1} \leftarrow R_i \setminus S$, $\mathcal{S}_{i+1} \leftarrow \mathcal{S}_i - \{S\}$.
}
\end{algorithm}

  \end{wrapfigure}

Our differentially private $k$-medians algorithm solves the Maximum Coverage problem as a subproblem.   The Maximum Coverage problem is defined as follows: on a universe $\mathcal{U}$ of items and a family $\mathcal{S}$ of subsets of $\mathcal{U}$, and a parameter $z$, the goal is to select $z$ sets in $\mathcal{S}$ to cover the most elements of $\mathcal{U}$. Formally, we are looking to find
\[\arg\max_{\mathcal{C}\subseteq \mathcal{S}, |\mathcal{C}| = z} \Bigl|\bigcup_{c\in\mathcal{C}}c\Bigr|.\]

Our approach for solving private Maximum Coverage is based on the Unweighted Set Cover algorithm in \cite{gupta2010differentially}. To preserve privacy, algorithm \ref{alg:setcover} chooses sets using the exponential mechanism, with probability related to the improvement in coverage caused by choosing the set.

Assume that there exists a selection of $z$ sets that covers $\mathcal{U}$.  A classic fact for the maximum coverage problem is that if we build the family $\mathcal{C}$ by always selecting the item in $\mathcal{S}\setminus\mathcal{C}$ that covers the largest number of uncovered elements in $\mathcal{U}$, then after $z$ iterations, $|\cup_{c\in\mathcal{C}}c| \geq (1- 1/e) |\mathcal{U}|$. Here we show an observation that will be useful for our algorithm later. The proof is in the appendix.
\begin{lemma}
\label{lem:2eps}
For $\epsilon > 0$, if we always select a set that covers at least half as many uncovered elements as the set that covers the most uncovered elements, then after $2z\ln 1/\epsilon$ iteration, \(\Bigl|\bigcup_{c\in\mathcal{C}}c\Bigr| \geq (1- \epsilon) | \mathcal{U}|\).
\end{lemma}

\section{Private k-medians}
\label{sec:privatekmedians}
Given a set of points $V$, a metric $d: V \times V \rightarrow \mathbb{R}$, a private set of demand points $D  \subseteq V$, and a value $k < |V|$, the objective of the $k$-medians problem is to select a set of points (centers) $F \subset V$, $|F| = k$ to minimize $\text{cost}(F) = \sum_{v \in D} d(v,F)$, where $d(v,F) = \min_{f \in F} d(v,f)$. Let $\Delta= \max_{u,v \in V} d(u,v)$ be the diameter of the metric space. We use $\epsilon$ as the approximation parameter for maximum coverage problem in lemma \ref{lem:2eps}, and $\epsilon_p$ and $\delta_p$ as privacy parameters. Let $B_t(v)$ be the ball of radius $t$ centered at $v$ i.e. the set of all points in the metric space within distance $t$ from $v$.

Our approach is based on the Maximum Coverage problem. One way to compute the clustering cost is by computing for every distance threshold $t$, the number of points within distance $t$ from the centers and integrating the counts from $0$ to the maximum distance. Thus, if for every threshold $t$ the number of points farther than $t$ from the our solution's centers is not much more than the number of points farther than $t$ from the optimal centers, then our cost is not much larger than the optimal cost. Thus, our algorithm goes through distance thresholds from small to large and tries to ``cover'' as many points as possible using fresh centers every time. For each threshold, we aim to cover $1-\epsilon$ times the number of points the optimal solution can cover. The result is that our clustering cost is not much larger than the optimal cost, albeit using more centers. Since we use exponentially growing thresholds, we only use $O(\log n)$ times more centers, which results in a small error due to privacy noise. The full algorithm is described in Algorithm~\ref{alg:kmedian}.

\begin{algorithm}[]
\caption{The $k$-medians algorithm}
\label{alg:kmedian}
\DontPrintSemicolon
\SetAlgoLined
\KwIn{ a set of points $V$, a private set of demand points $D \subseteq V$, a metric $d$, $\epsilon$, $\epsilon_{p}$, $\delta_p$}

$V'=D, C = \emptyset, r = \lc 1+\log_{1+\epsilon}n \rc$

\For{$i$ from $1$ to $r$}{
Set $t_i = (1+\epsilon)^{i-1}\Delta/n$

Run algorithm \ref{alg:setcover} for the set system $\mathcal{U}=V, \mathcal{S}= \{B_{t_i}(v) \cap V': v \in V\}$, with $\epsilon_{s} = \frac{\epsilon_p}{2}, \delta_s = \delta_p$ for  $m = 2k\ln(1 \slash \epsilon) $  iterations to get $C_i$

$V_i = \bigcup_{v\in C_i} B_{t_i}(v) \cap V'$

$V' = V' \setminus   V_i $

$C = C \cup C_i$
}

Assign each point in $D$ to its closest point $c \in C$ 

Let $n_c$ be the total number of points assigned to $c$ for each $c\in C$

For each $c\in C$, set $n_c' = n_c + Y_c$ where $Y_c \sim \text{Lap}\left(\frac{2}{\epsilon_p} \right)$ 

Run a k-medians algorithm to select $k$ centers from $V$, with demand points at each $c\in C$ with multiplicity $n_c'$
\end{algorithm}

The algorithm begins with the discretization of distance thresholds.  The goal is similar to our discussion above. We apply algorithm \ref{alg:setcover} to select a set of points that cover a large set of demand points across different distance thresholds. Note that the objective cost of any set of center following the discretization scheme is not too far from the actual costs, as we will show in lemma \ref{lem:approx}. Thus, the set of centers that we find across different thresholds $t$ should also have cost similar to $\OPT$. 

Our final step is to obtain the final set of centers from the potential centers. To preserve privacy for this step, we create a new dataset similar to the original one with some privacy. In this new dataset, every demand point is shifted towards the closest center from the previous step, and we apply the Laplace mechanism to the assigned number of demands points of each center to preserve privacy. At this point, we can use a $k$-medians algorithm on this dataset to output the final $k$ centers. In this new problem, although the objective changes because of shifting and the Laplace mechanism, the point that can be selected to be centers are the same as before. Thus, the cost of the centers returned in the final step is at most the the cost of this new objective plus the total shifting distance.  In the following sections, we will first analyze the privacy and then the utility of this algorithm.

\subsection{Privacy Analysis}
We first show that this algorithm is $(\epsilon_p, \delta_p)$ differentially private.
To show this, we first show that the entirety of the for loop is $(\epsilon_p / 2, \delta_p)$ differentially private, and then take advantage of composition and apply lemma \ref{thm:laplacemech} at line 12 to obtain the final result. Note that the analysis of the loop very closely follows the proof for privacy of Unweighted Set Cover in \cite{gupta2010differentially}; their algorithm selects sets in a particular order to form a cover, while our algorithm selects candidate centers with increasing distance thresholds, where a center is assumed to cover all demand points within its distance threshold. The significant difference in the two proofs is that our algorithm could select the same center twice with different distance thresholds while a set will never be chosen twice in set cover in \cite{gupta2010differentially}. This re-selection of the same center will not affect the differential privacy of the algorithm, because the privacy analysis hinges on which demand points have been covered, not which centers have been selected. As a result of this, we save an additional $\log n$ factor on privacy, which removes a $\log n$ term from the additive error in lemma \ref{lem:coverage} which carries through the additive error in the utility. We include the proof in the appendix and omit it here, due to its close similarity to \cite{gupta2010differentially}.

\begin{lemma} \label{lem:kmediansloopprivacy}
The for loop in algorithm \ref{alg:kmedian} preserves $(\epsilon_p / 2, \delta_p)$ differential privacy.
\end{lemma}

The function affected by the Laplace mechanism in line 11 returns a vector of the counts $n_c$. In the case of sets $A, A'$ as in theorem \ref{thm:laplacemech}, the difference between $f(A)$ and $f(A')$ is exactly one for one item in this vector, and therefore $\|f(A) - f(A')\|$ = 1, so the function has $L_1$ sensitivity 1. Thus, line 11  is $(\epsilon_p / 2, 0)$ differentially private by theorem \ref{thm:laplacemech}. By composition, this fact and lemma \ref{lem:kmediansloopprivacy} yield the following lemma:

\begin{lemma}
\label{lem:kmedianprivacy}
Algorithm \ref{alg:kmedian} is $(\epsilon_p, \delta_p)$ differentially private.
\end{lemma}

\subsection{Utility Analysis}

We define $t_0 = 0, t_1 = \frac{\Delta}{n},  t_2 = \frac{\Delta(1+\epsilon)}{n},..., t_r = \Delta$ as shorthand for the thresholds. Also, let $o_i$ be the number of points at distance in the range $[t_{i-1}, t_i)$ from their center in the optimal solution (which we denote $\OPT$), and let $a_i$ be the the number of points at distance in the range $[t_{i-1}, t_i)$ from their closest point in $C$ after the for-loop in algorithm \ref{alg:kmedian}. To bound the performance of our solution, we first show that discretizing the distance thresholds at $t_i$'s instead of integrating from $0$ to $\Delta$ introduces negligible error to the cost of the solution (see Lemma~\ref{lem:approx}). Next, for each distance threshold, lemma \ref{lem:coverage} uses the approximation guarantee of maximum coverage to show that we are efficiently covering demand points using not many more centers than $\OPT$. Crucially, lemma \ref{lem:solutionapprox} shows that by covering almost as well as $\OPT$ at every distance threshold, our solution has cost not much more than that of $\OPT$.

The following lemmas bound $\OPT$ by the threshold approximation and then bound our solution by that approximation respectively.

\begin{lemma}
\label{lem:approx}
$\sum_{i=1} o_i t_i \leq (1+\epsilon) \OPT + \Delta$
\end{lemma}

\begin{proof}
On all $u \in D$ and set of centers $F$, define $d'(u,F)$ as the minimum distance threshold $t_i$ which is larger than $d(u,F)$. If $d(u,\text{OPT}) > \frac{\Delta}{n^2}$ then $d'(u,\text{OPT}) \leq (1+\epsilon) d(u,\text{OPT}) $. If $d(u,\text{OPT}) \leq \frac{\Delta}{n^2}$, then $d'(u,\text{OPT}) = \frac{\Delta}{n} \leq d(u, \text{OPT}) + \frac{\Delta}{n}$. Summing over $d \in D$ yields the bound.
\end{proof}

For each distance threshold $t_i$, the following lemma shows that the algorithm covers almost as many points as $\OPT$. Its proof is left in the appendix. The result follows mostly from lemma \ref{lem:2eps} and the error of the exponential mechanism in algorithm \ref{alg:setcover}.

\begin{lemma}
\label{lem:coverage}
Consider iteration $i$ of the for-loop and let $M_i$ be the maximum coverage of $k$ centers with radius $t_i$ over points in $V'$.  With high probability, at line 5,  $|V_i| \geq (1-\epsilon)M_i - \frac{24k \ln n \ln \left(\frac{e}{\delta_p}\right)}{\epsilon_p}$.
\end{lemma}

The next lemma relates the cost of our solution and that of $\OPT$ given that we cover almost as well as $\OPT$ at every distance threshold.

\begin{lemma}
\label{lem:solutionapprox}
$ \sum_{i=1}^r a_i t_i \leq   \frac{1-\epsilon}{1-\epsilon - \epsilon^2}  \sum_{i=1} o_i t_i +   \frac{24\Delta k \ln n \ln \left(\frac{e}{\delta_p}\right)}{\epsilon_p (1-\epsilon -\epsilon^2)}  $
\end{lemma} 

\begin{proof}
Let $O_i = \sum_{j=i}^r o_j, A_i = \sum_{j = i}^r a_j$ and $ E = \frac{24k \ln n \ln \left(\frac{e}{\delta_p}\right)}{\epsilon_p}$. Given a threshold $t_i$, we know that the centers in OPT cover $n-O_{i+1}$ points with distance at most $t_i$. At threshold $t_i$, algorithm \ref{alg:kmedian} has already covered $n-A_{i}$ points so we know that there is a solution covering at  additional $(n-O_{i+1}) - (n-A_{i}) = A_{i}-O_{i+1}$ points. By the guarantee of the greedy set cover algorithm in lemma \ref{lem:coverage}, we  cover $a_i \geq (1-\epsilon) (A_{i} - O_{i+1}) - E$ new points on the next iteration. By substituting $A_i = a_i + A_{i+1}$, we have $a_i \geq \frac{(1-\epsilon)}{\epsilon} (A_{i+1}-O_{i+1}) - \frac{E}{\epsilon}$. Notice that: 
\begin{align*}\sum_{i=1}^ra_i t_i= \sum_{i=1}^r A_i (t_i - t_{i-1}) \leq  \sum_{i=1}^r \left(\frac{\epsilon a_{i-1}}{1-\epsilon} + O_i + \frac{E}{1-\epsilon}\right) (t_i - t_{i-1}) &\\=  \sum_{i=1}^r \frac{\epsilon a_{i-1}}{1-\epsilon}  (t_i - t_{i-1}) + \sum_{i=1}^r o_i t_i + \frac{\Delta E}{1-\epsilon}
\end{align*}
The equality is because of the telescoping sums. Also notice that:
\begin{equation*} \sum_{i=1}^r \frac{\epsilon a_{i-1}}{1-\epsilon}  (t_i - t_{i-1}) = \sum_{i=1}^{r-1} \frac{\epsilon a_{i}}{1-\epsilon}  (t_{i+1} - t_{i}) \leq \sum_{i=1}^{r-1} \frac{\epsilon^2 }{1-\epsilon} a_{i} t_{i}  \end{equation*}

The last inequality is because for all $1 \le i \le r-1$, $t_{i+1} = (1+\epsilon)t_i$ by definition. We are also able to drop the term $i = 0$ from the last two sums because $a_0 = 0$.

Thus:
\begin{align*}
   & \sum_{i=1}^r a_i t_i - \sum_{i=1}^{r-1} \frac{\epsilon^2 }{1-\epsilon} a_{i} t_{i} = \frac{1-\epsilon - \epsilon^2}{1-\epsilon}   \sum_{i=1}^{r-1} a_i t_i + a_rt_r  \leq  \sum_{i=1} o_i t_i + \frac{\Delta E}{1-\epsilon} \\ & \implies  \sum_{i=1}^r a_i t_i \leq   \frac{1-\epsilon}{1-\epsilon - \epsilon^2}  \sum_{i=1} o_i t_i + (1-\epsilon -\epsilon^2) \Delta E
\end{align*}
\end{proof}

Combining the results of lemmas \ref{lem:approx} and \ref{lem:solutionapprox}, we see that 
\[\sum_{i=1}^r a_it_i \le \frac{1-\epsilon}{1-\epsilon - \epsilon^2}\left((1+\epsilon)\OPT + \Delta\right) + (1-\epsilon - \epsilon^2)\Delta E\]
which gives us a bound on the cost of snapping points in $D$ to points in $C$.

\begin{lemma}
\label{lem:prelimcost}
Consider the k-medians problem in the last line of algorithm \ref{alg:kmedian}, where demand points in $D$ are shifted to points in $C$ and Laplace noise is applied. With high probability, the optimal objective cost of this new k-medians problem is at most
 \begin{equation*}
     \OPT + \sum a_i t_i + \frac{4\Delta k\ln(1/\epsilon)}{\epsilon_p}\left(\frac{\ln(n)}{\ln(1+\epsilon)} + 2\right) 
 \end{equation*}
where $\OPT$ is the cost of the original k-medians problem.
\end{lemma}

We leave the details of this proof in the appendix. Briefly, we can choose OPT of the original problem as a solution to the new problem. In this case, the original points get the cost of OPT plus the cost of snapping to the centers in $C$, which is $\sum a_i t_i$. The last term is an upper bound on the error from using the Laplace mechanism, with high probability.

Using an approximation algorithm for the k-medians problem with approximation factor $M$ in the last step of algorithm \ref{alg:kmedian}, the entire cost gains a multiplicative factor $M$. Therefore, we can summarize the utility of algorithm \ref{alg:kmedian} into the following lemma and even simpler theorem:

\begin{lemma}
\label{lem:utility}
With high probability, algorithm \ref{alg:kmedian} preserves $(\epsilon_p, \delta_p)$ differential privacy and solves the k-medians problem with cost
\[O\left(M(1+\epsilon)\right)\OPT + O\left(\frac{Mk\Delta}{\epsilon_p}\ln n \left(\ln\left(\frac{e}{\delta_p}\right) + \frac{\ln(1/\epsilon)}{\ln(1+\epsilon)} \right) \right)\]
where the black-box k-medians algorithm used in the last step of algorithm \ref{alg:kmedian} has approximation factor $M$ and $\epsilon$ is a small, positive constant.
\end{lemma}

The full proof of this lemma can be found in the appendix. In short, we merge the results of lemmas \ref{lem:approx}, \ref{lem:solutionapprox}, and \ref{lem:prelimcost},
and bound the final cost using the private problem's cost, the snapping cost, and the triangle inequality.

\begin{corollary}
By using a constant approximation non-private algorithm for $k$-medians, there is a $(\epsilon_p, \delta_p)$-differentially private algorithm for the $k$-medians problem that, with high probability, outputs a solution at cost \[O\left(\OPT + \frac{k\Delta}{\epsilon_p} \ln (n)  \left(\ln\left(\frac{1}{\delta_p}\right) + \frac{\ln(1/\epsilon)}{\ln(1+\epsilon)} \right) \right). \]
\end{corollary}
Note that our algorithm allows for any $k$-medians algorithm to be used at the last step. One can choose a preferred trade-off between runtime and performance to select a suitable algorithm. This is in contrast to the approach in \cite{gupta2010differentially}, where the algorithm builds on the $k$-median algorithm in \cite{arya2004local}. 
\bibliography{kmedian.bib}

\section{Application to Euclidean k-means} 
\label{sec:euclidean}

In the Euclidean $k$-medians problem, instead of having a discrete set of demand points, $V$ is defined to be all of $\mathbb{R}^d$. We wish to select a set of points (centers) $F \subset \mathbb{R}^d$, $|F| = k$ to minimize $\text{cost}(F) = \sum_{v \in D} d(v,F)^2$. In this section, we will apply our result to improve additive error in the approach in \cite{stemmer2018differentially}.

The strategy in \cite{stemmer2018differentially} is to first identify a polynomial set of candidate centers such that it contains a subset of $k$ candidate centers with low $k$-means cost. Then, the algorithm uses a private discrete $k$-means algorithm to select the final $k$-centers with low cost from the set of candidate center. More concretely, the algorithms in \cite{stemmer2018differentially} is guaranteed to output a $(\epsilon_p, \delta_p)$-private set of candidate centers $Y$ of size at most $\epsilon_p n \log(\frac{k}{\beta})$ such that with probability at least $1 - \beta $, there exist a subset of size $k$ center with constant multiplicative error and additive error of $O(T^{\frac{1}{1-a-b}} \cdot w^{\frac{1}{1-a-b}} \cdot k^{\frac{1}{1-a-b}}) \Delta^2$, where $a,b$ are small constant parameters of the Locality Sensitive Hashing algorithm used in \cite{stemmer2018differentially}, $T = \Theta(\log \log n)$ and  $$w = O\left(\frac{ \sqrt{d}}{\epsilon_p} \cdot \log\log n \cdot \log\left(\frac{k}{\beta}\right) \sqrt{\log \frac{\log \log n}{\delta_p}}\right)$$
Note that $a,b$ can be chosen arbitrarily small at the cost of making the multiplicative approximation factor a larger constant. The work of~\cite{stemmer2018differentially} focuses on the regime where $a,b$ are small and $1/(1-a-b) \le 1.01$. The resulting additive error for identifying candidate centers is $\tilde{O}_{\epsilon_p, \delta_p}(k^{1.01}d^{.51})$. 

The performance bottleneck of the Euclidean $k$-means is in the algorithm to select the final $k$ centers from a candidate set. We can apply our algorithm on the potential center returned by $\cite{stemmer2018differentially}$ to improve the algorithm performance. Note that our algorithm can be applied to solve the $k$-means objective by passing the correct distance function. Although the squared root of Euclidean distance is not a metric, we can still apply the same algorithm and analysis to get the bound for $k$-means objective. The only difference is at the last step, instead of running a $k$-medians algorithm on the returned center, we run a $k$-means algorithm to get final centers.

Rather than replicate the entire proof, we will only review the sections which are affected by the change in the distance function. Furthermore, we will extend the proof to all distance functions $d^p$ for any natural number $p$.

Notably, the privacy analysis is independent of the distance function, and is therefore unaffected. In fact, the only steps in the proofs of section \ref{sec:privatekmedians} which involve the distance function are lemma \ref{lem:approx} and lemmas \ref{lem:prelimcost} and \ref{lem:utility}. For the distance function $d^p$, lemma \ref{lem:approx} is amended as follows:

\begin{lemma}
\label{lem:approxk}
$\sum_{i=1} o_i t_i \leq (1+\epsilon)^p \text{OPT} + \Delta$
\end{lemma}
\begin{proof}
In the case where $(d(u,\text{OPT}))^p > \frac{\Delta}{n^2}$, now $(d'(u,\text{OPT}))^p \leq (1+\epsilon)^p (d(u,\text{OPT}))^p $. Otherwise, the proof is functionally identical to that of lemma \ref{lem:approx}.
\end{proof}

For lemmas \ref{lem:prelimcost} and \ref{lem:utility}, we directly address and resolve the main issue the distance function faces here, which is that the triangle inequality does not hold when $p > 1$. However, we can use the following lemma, which we prove in the appendix:
\begin{lemma}
\label{lem:amendedtriangleineq}
In any metric space and $p\ge 1$,
$(d(a,b))^p \le 2^{p-1}\left((d(a,c))^p + (d(b,c))^p\right)$.
\end{lemma}

Using this lemma, we see that when the triangle inequality would be applied, we gain an additional $2^{p-1}$ constant. This affects the leading constant of the approximation factor and additive error, but does not affect the asymptotic cost.

When $p = 2$, this is the distance function used in the $k$-means problem. Therefore, the only changes to $\ref{alg:kmedian}$ necessary to make it a functional $k$-means algorithm are to use the proper input metric $d$ and to run a black-box $k$-means algorithm as the last step rather than a black-box $k$-medians algorithm. Then, if we are trying to minimize the objective function using the distance function $d^p$ and we use a black-box algorithm for this objective function at the last step of algorithm $\ref{alg:kmedian}$, the proofs in section \ref{sec:privatekmedians} using lemma \ref{lem:approxk} instead of \ref{lem:approx} yields the following lemma and corollary:

\begin{lemma}
\label{lem:kforallL}
Given a problem equivalent to $k$-means but with distance function $d^p$, and a discrete algorithm for that problem with approximation factor $M$, there exists an $(\epsilon_p, \delta_p)$ differentially-private algorithm for that problem which, with high probability, has objective cost at most
\[O\left(2^{p-1}M(1+\epsilon)^p\right)\OPT + O\left(\frac{Mk\Delta^p}{\epsilon_p}  \left(\ln n  +  2^{p-1}\ln|Y|\ln\left(\frac{e}{\delta_p}\right)\right)\right)\]

\end{lemma}

\begin{corollary}
There is a $(\epsilon_p,\delta_p)$-differentially private algorithm for the Euclidean $k$-means problem that with probability at least $1-\beta$ returns a solution with a constant multiplicative factor and an additive error of 
\begin{equation*}
    O\left(\Delta^2\left(T^{\frac{1}{1-a-b}} \cdot w^{\frac{1}{1-a-b}} \cdot k^{\frac{1}{1-a-b}} \right)  + \frac{\Delta^2  k}{\epsilon_p}\left(\ln n + \ln\left(\epsilon_p n \log\left(\frac{k}{\beta}\right)\right)\ln\left(\frac{e}{\delta_p}\right) \right) \right)
\end{equation*}

\end{corollary}

Note that our algorithm results in better additive term compared to applying \cite{gupta2010differentially} on the potential centers. Specifically, the second additive term is almost linear in $k$ instead of $k^{1.5}$, making the entire additive error almost linear in $k$. 

\section{Broader Impact Statement}
Clustering has many applications in machine learning, such as image segmentation \cite{yang2017towards, patel2013latent}, collaborative filtering \cite{mcsherry2009differentially, schafer2007collaborative}, and time series analysis \cite{mueen2012clustering}. Privacy is a major concern when input data contains sensitive information.  Differential privacy \cite{laplacenoise}  has become a rigorous framework for ensuring privacy in algorithms. Thus, differentially private algorithms for clustering problem would ensure for each individual in the input a robust privacy guarantee.

Our improved utility guarantee will perhaps encourage adoption of privacy-preserving algorithm as a replacement of the non-private counterpart. Furthermore, our approach allows for usage with other clustering algorithm as a black-blox. We believe this further improves the applicability of private clustering algorithms, making it easier to incorporate privacy guarantee into existing clustering frameworks. The limitations of the work are that the privacy guarantee requires certain assumptions on the input data such as the data being bounded and the utility guarantee has additive error that is only meaningful when the dataset has a large enough number of participants. When applying the algorithm, the curator has to ensure that the assumptions hold to protect the privacy of the participants.

\bibliographystyle{acm}

\newpage
\appendix
\section{Missing proofs}
\begin{theorem}
\label{theorem:concentration}
If $Z_1, Z_2,...Z_n$ are i.i.d random variables that follow exponential distrubtion with parameter $\lambda > 0$, then $\Pr[\sum_{i=1}^nZ_i \geq 2n\lambda] \leq \left(\frac{2}{e}\right)^n$
\end{theorem}
\begin{proof}
Recall that the moment generating function of $Z_i$ is
    $\E[\exp(tZ_i)] =\frac{\lambda}{\lambda-t}$ for $t < \lambda $ and the fact that $\E\left[\exp\left(t \sum_{i=1}^nZ_i \right)\right] = \left(\frac{\lambda}{\lambda-t}\right)^n $ based on the property of moment generating function. We have:
\begin{align*}
    \Pr\left[\sum_{i=1}^nZ_i \geq 2n/\lambda\right] = \Pr\left[\exp\left(t \sum_{i=1}^nZ_i \right) \geq \exp\left(2tn/\lambda\right)\right] &\leq \E\left[\exp\left(t \sum_{i=1}^nZ_i \right)\right] / \exp(2tn/\lambda) \\ &=  \left(\frac{\lambda}{(\lambda-t)\exp(2t/\lambda)}\right)^n  
\end{align*}

The inequality follows from Markov's inequality. Let $t=\lambda / 2$, we have the theorem. 
\end{proof}

\begin{lemma}
\label{lem:2eps_app} \textbf{(Lemma 3.5)}
For $\epsilon > 0$,  if  we always select a set covers at least half of the number of elements in the set that covers the most uncovered elements, then after $2z\ln 1/\epsilon$ iteration, $\Bigl|\bigcup_{c\in\mathcal{C}}c\Bigr| \geq (1- \epsilon) | \mathcal{U}|$.
\end{lemma}
\begin{proof}
Define $U_j$ to be the number of uncovered elements after $j$ iterations and $A_j$ to be the number of newly covered element at iteration $j$. We will first show that $U_{j+1} \leq (1-\frac{1}{2z})^{j+1}|\mathcal{U}|$. We will prove by induction.  For the base case $j=0$, We know for each iteration there exists some set that can cover at least $1/z$ of the remaining uncovered elements. Otherwise, it's not possible to cover all elements with only just $z$ steps by the optimal solution. Thus, because we always select  at least half of the number of elements in the set that covers the most uncovered elements, $A_j \geq  \frac{U_{j-1}}{2z}$. This proves the base case. Now for the inductive hypothesis, assume that $U_{j} \leq (1-\frac{1}{2z})^{i}|\mathcal{U}|$. Because of the base case we have: 
\begin{align*}
    U_{j+1} \leq U_{j} - U_j /2z  \implies
    U_{j+1}  \leq   (1-\frac{1}{2z})^{i}|\mathcal{U}| (1-\frac{1}{2z}) =  (1-\frac{1}{2z})^{i+1}|\mathcal{U}|
\end{align*}
where the last inequality we used the induction hypothesis. 

Now, note that if there  are $z \cdot i$ iterations : \begin{equation*}U_{z \cdot i} \leq  \left[(1- \frac{1}{2z}) \right]^{z\cdot i} | \mathcal{U}|  \leq \left(\frac{1}{e}\right)^{i/2} |\mathcal{U}| \end{equation*}
For $i \geq 2 \ln\frac{1}{\epsilon}$, observe that the RHS is at most $(1-\epsilon)|\mathcal{U}|$, as desired.
\end{proof}

\begin{lemma} \label{lem:kmediansloopprivacy_app}
\textbf{(Lemma 4.1)}
The for loop in algorithm \ref{alg:kmedian} preserves $(\epsilon_p / 2, \delta_p)$ differential privacy.
\end{lemma} 

\begin{proof}
  Let $A$ and $B$ be two sets of demand points from the same point set $V$, such that $A$ and $B$ have a symmetric difference to be a single element $I \in V$. For the sake of this proof, we care about the order in which points from $V$ were selected by the iterations of algorithm \ref{alg:setcover}. Furthermore, the same point $v \in V$ may be selected in multiple iterations of algorithm \ref{alg:setcover}, by being chosen at different thresholds. In order to consider these as separate instances, we will record the order of selected points with the current threshold, as (point, threshold) pairs. We will denote the order of (point, threshold) pairs selected on input $A$ as $\pi_A$, and on input $B$ as $\pi_B$. Note that on any set of demand points, the total number of selections is the same for each iteration of algorithm \ref{alg:setcover}, so when checking equality of two selections $\pi_1$ and $\pi_2$ from the same set of points $V$, the thresholds in the pairs should always match, and equality depends only on the points in the pairs.
  
  First, we want to fix some order $\pi$ and bound the ratio between $\Pr[\pi_A = \pi]$ and $\Pr[\pi_B = \pi]$. Fix $\pi$. We denote by $h$ the total selections in all calls to algorithm \ref{alg:setcover} in the for loop. Note that this value $h$ is independent of $D$ given $V$, and is therefore equivalent for both $A$ and $B$. We will also write $s_{i,j}(D)$ to denote the size of $\mathcal{S}_j$ in algorithm \ref{alg:setcover} (which is the updated value of $V'$ after each iteration of the for loop) after $i-1$ sets have been added to the cover when the demand set is $D$. Then, the probability we make the same choice as $\pi$ on iteration $i$ given that we made the same choices as $\pi$ on the first $i - 1$ iterations is given by \[\frac{\exp{(\epsilon'\cdot s_{i,\pi_i}(D))}}{ \sum_j\exp{(\epsilon'\cdot s_{i,j}(D))}}\]
  where the numerator is the relative probability that algorithm \ref{alg:setcover} picks $\pi_i$, and the denominator is the sum of relative probabilities of all possible choices $j$. Therefore,
  \begin{align*}
    \frac{\Pr[M(A)=\pi]}{\Pr[M(B)=\pi]} 
    &= \prod_{i=1}^h\left(\frac{\exp(\epsilon'\cdot s_{i,\pi_i}(A))/(\sum_j\exp(\epsilon'\cdot s_{i,j}(A)))}{\exp(\epsilon'\cdot s_{i,\pi_i}(B))/(\sum_j\exp(\epsilon'\cdot s_{i,j}(B)))}\right) \\
    &= \frac{\exp(\epsilon'\cdot s_{t,\pi_t}(A))}{\exp(\epsilon'\cdot s_{t,\pi_t}(B))}\cdot \prod_{i=1}^t\left(\frac{\sum_j\exp(\epsilon'\cdot s_{i,j}(B))}{\sum_j\exp(\epsilon'\cdot s_{i,j}(A))}\right)
  \end{align*}
  where $t$ is such the first set in $\pi$ containing $I$ is $S_{\pi_t}$. After $t$ choices, the remaining demand points for $A$ and $B$ are identical since $I$ is covered, and therefore all the following probabilities are equivalent and the terms cancel to multiplicative factor 1. Also, since the items in $\pi$ before $\pi_t$ do not contain $I$, the relative probabilities of choosing those items is the same, so the numerators cancel for those indices.
  
  If $A$ contains $I$ and $B$ does not, then the first term is $\exp(\epsilon')$ since $s_{t,\pi_t}(A) = s_{t,\pi_t}(B) + 1$, and each term in the product is at most 1, since $s_{i,j}(A) \ge s_{i,j}(B)$. Therefore, the whole term is at most $\exp(\epsilon')$. Since $\epsilon' \le \epsilon_s$ for $\delta_p = \delta_s \le 1$, it follows that the for loop is $(\epsilon_p / 2, 0)$ differentially private for this case, which shows the weaker $(\epsilon_p / 2, \delta_p)$ differential privacy.

  Now, suppose $B$ contains $I$ and $A$ does not. In this case, the first term is $\exp(-\epsilon') < 1$. In instance $B$, every (point, threshold) pair which contains $I$ covers exactly 1 more item than that pair in $A$, and all others remain the same size. We denote the set of such pairs as $S^I$.
  Therefore, we have:
  \begin{align*}
    \frac{\Pr[M(A)=\pi]}{\Pr[M(B)=\pi]} & \leq  \prod_{i=1}^t \left(\frac{(\exp(\epsilon')-1)\cdot\sum_{j\in S^I}\exp(\epsilon'\cdot s_{i,j}(A)) + \sum_j\exp(\epsilon'\cdot s_{i,j}(A))}{\sum_j\exp(\epsilon'\cdot s_{i,j}(A))}\right) \\
    &  =  \prod_{i=1}^t \left(1 + (\exp(\epsilon')-1)\cdot p_i(A) \right) \\
    &  \le \prod_{i=1}^t \exp ((\exp (\epsilon') - 1) \cdot p_i(A))
  \end{align*}
  where $p_i(A)$ is the probability that a set containing $I$ is chosen at step $i$ of the algorithm running on instance $A$, conditioned on picking the sets $S_{\pi_1}, \ldots, S_{\pi_{i-1}}$ in the previous
  steps. The last step follows because $1+x \le \exp{(x)}$ when $x \ge 0$.

For an instance $A$ and an element $I\in A$, we say that an order of chosen (point, threshold) pairs $\sigma$ is \textit{$q$-bad} if the sum $\sum_{i} p_i(A) \textbf{1}(I \text{ uncovered at step } i)$ is larger than $q$, where $p_i(A)$ is as defined above. We call $\sigma$ \textit{$q$-good} if it is not $q$-bad. We first consider the case when $\pi$ is $(\ln \delta_p^{-1})$-good. Since the index $t$ corresponds to the first set in $\pi$ containing $I$, we have
 \begin{equation*}
     \sum_{i=1}^{t-1}p_i(A) \leq \ln \delta_p^{-1}.
  \end{equation*}
Continuing the earlier analysis,
  \begin{align*}
    \frac{\Pr[M(A)=\pi]}{\Pr[M(B)=\pi]} &\leq
    \prod_{i=1}^t \exp((\exp(\epsilon')-1)p_i(A))
    \leq \exp(2\epsilon'\sum_{i=1}^tp_i(A)) \\
    &\leq \exp(2\epsilon'(\ln(\frac{1}{\delta_p})+p_t(A)))
    \leq \exp(2\epsilon'(\ln(\frac{1}{\delta_p})+1)) \le \exp{(\epsilon_p / 2)}.
  \end{align*}
Thus, for any $(\ln \delta_p^{-1})$-good output $\pi$, we have $\frac{\Pr[M(A)=\pi]}{\Pr[M(B)=\pi]} \leq \exp(\epsilon_p / 2)$.
We can then invoke the following lemma, which is posed as Lemma 6.4 and proved in appendix B in \cite{gupta2010differentially}.
\begin{lemma} \label{lem:deltafix}
For any instance $A$ and any $I \in A$, the probability that the output $\pi$ is $q$-bad is bounded by $\exp(-q)$.
\end{lemma}

Thus for any set $\mathcal{P}$ of series of choices, we have
\begin{align*}
&\Pr[M(A) \in \mathcal{P}]\\
&= \sum_{\pi \in \mathcal{P}} \Pr[M(A)=\pi]\\
&= \sum_{\pi \in \mathcal{P}: \pi\mbox{ is } (\ln \delta_p^{-1})\mbox{-good}} \Pr[M(A)=\pi] + \sum_{\pi \in \mathcal{P}: \pi\mbox{ is } (\ln \delta_p^{-1})\mbox{-bad} } \Pr[M(A)=\pi] \\
&\leq \sum_{\pi \in \mathcal{P}: \pi\mbox{ is } (\ln \delta_p^{-1})\mbox{-good} } \exp(\epsilon_p / 2) \Pr[M(B)=\pi] + \delta_p\\
&\leq \exp(\epsilon_p / 2)\Pr[M(B) \in \mathcal{P}] +\delta_p.
 \end{align*}

Therefore, we have shown $(\epsilon_p / 2, \delta_p)$ differential privacy in both cases.
\end{proof}

\begin{lemma}
\label{lem:coverage_app}
\textbf{(Lemma 4.4)}
Consider iteration $i$ of the for-loop and let $M_i$ be the maximum coverage of $k$ centers with radius $t_i$ over points in $V'$.  With high probability, at line 5,  $|V_i| \geq (1-\epsilon)M_i - \frac{24k \ln n \ln \left(\frac{e}{\delta_p}\right)}{\epsilon_p}$.
\end{lemma}

\begin{proof}
The items in the family $\mathcal{S}$ in line 4 are exactly one-to-one with the points in $V$, and the set $s_v \in \mathcal{S}$ corresponding to $v\in V$ are exactly the points which are not within distance $t_{i-1}$ of an existing center in $C$ but are within distance $t_i$ of $v$. Therefore, the items covered by the centers in $C_i$ are all within distance $t_i$ of their closest centers, so the change in the coverage of $C$ is at least the size of the set coverage from $C_i$.

Our analysis is similar to \cite{gupta2010differentially}. The main difference is that instead of covering all points that $\OPT$ can cover, we aim to cover a $(1-\epsilon)$ portion within an additive error, hence we run $2k \log(1 / \epsilon)$ iterations instead of $2k \log n$. Consider $|R_i|$ to be the number of remaining elements yet to be covered, and define $L_i =\max_{S \in \mathcal{S}} |S \cap R_i|$, the largest number of uncovered
elements covered by any set in $\mathcal{S}$.

By theorem \ref{theorem:expLoss}, the exponential mechanism when selecting set ensures that with probability at most $1 / n^2$ that we select a center with coverage less than $L_i - \frac{3 \ln n}{\epsilon'} $. When $L_i > \frac{6 \ln n}{\epsilon'} $, we are guaranteed to choose a center that covers at least $L_i / 2$ points. Based on lemma \ref{lem:2eps}, for each iteration, as long as $L_i > \frac{6 \ln n}{\epsilon'}$ we always takes the greedy option and are guaranteed to have  $|C_i| \geq (1-\epsilon)M_i$ with probability at least $1 - \frac{1}{n}$. However, when $L_i \leq  \frac{6 \ln n}{\epsilon'} $ , although we are not guaranteed to take the greedy action, there are at most $\frac{6k \ln n}{\epsilon'}$  yet to be covered by the algorithm compared to $\OPT$ at radius $r$. Thus, the algorithm loses at most  $\frac{6k \ln n}{\epsilon'}$ points.
\end{proof}

\begin{lemma}
\label{lem:prelimcost_app}
\textbf{(Lemma 4.6)}
Consider the k-medians problem in the last line of algorithm \ref{alg:kmedian}, where demand points in $D$ are shifted to points in $C$ and Laplacian noise is applied. With high probability, the optimal objective cost of this new k-medians problem is at most
 \begin{equation*}
     \OPT + \sum a_i t_i + \frac{4\Delta k\ln(1/\epsilon)}{\epsilon_p}\left(\frac{\ln(n)}{\ln(1+\epsilon)} + 2\right)
 \end{equation*}
where $\OPT$ is the cost of the original k-medians problem.
\end{lemma}

\begin{proof}
After assigning every point in $D$ to the closest points in $C$, we run a $k$-medians algorithm on a multiset defined by $C$, where each element $c \in C$ has multiplicity  $n'_c$ as in line 11 of algorithm \ref{alg:kmedian}.  Recall that the absolute value of a random variable following Laplace distribution with parameter $b$ follows an exponential distribution with parameter $\frac1b$. Also, by theorem \ref{theorem:concentration}, the sum of exponential variables will be less than twice the expectation with high probability. For the sake of completeness, we include the proof of this fact in theorem \ref{theorem:concentration}. It is also significant that, since we call algorithm \ref{alg:setcover} a total $\lc \log_{(1+\epsilon)}n + 1 \rc = \lc \frac{\ln n}{\ln(1+\epsilon)} + 1 \rc$ times with $m = 2k\ln(1/\epsilon)$, we select at most $|C| \le 2k\ln(n)\ln(1/\epsilon) / \ln(1+\epsilon) + 2k\ln(1/\epsilon)$ centers before calling the black-box k-medians algorithm. With all this preliminary information, we begin to prove the claim.

The last term in the bound is obtained by the Laplace mechanism, where noise is applied to the counts of each center $c\in C$. Each of the  centers in $C$ has Laplacian noise applied to it with parameter $2/\epsilon_p$. Therefore, with high probability at most $\frac{4k\ln(1/e)}{\epsilon_p}\left(\frac{\ln(n)}{\ln(1+\epsilon)} + 2\right)$ demand points are "added" by line 11 of the algorithm, and each of these points is at distance at most $\Delta$ from their closest center, which yields the last term in the bound.

The first two terms come from the fact that we shift points in line 9 of algorithm \ref{alg:kmedian} and the triangle inequality. For each of the original demand points $v \in D$, let the cluster center in OPT closest to $v$ be $OPT_v$, and let the point in $C$ closest to $V$ be denoted $c_v$. We see that
\[d(c_v,OPT_v) \le d(v,c_v) + d(v, OPT_v)\]
by the triangle inequality. Since OPT of the original k-medians problem is a candidate solution for the new k-medians problem, the objective cost of using OPT upper bounds the optimal cost. Summing over all $v\in D$, we see that the objective cost of shifted demand points is therefore bounded as
\[\sum_{v\in D}d(c_v, OPT_v) \le \sum a_it_i + \OPT\]
since $\sum a_it_i$ is an upper-bound approximation of the cost of shifting the demand points to centers in $C$, and the sum of $d(v,OPT_v)$ yields exactly $\OPT$. Thus, the first two terms come from the cost of shifting the real demand points, and the last term comes from the Laplacian noise.

Note that there may be a better solution than the original k-medians' OPT to this new k-medians problem, but this is consistent with the bound by inequality.
\end{proof}

\begin{lemma}
\label{lem:utility_app}
\textbf{(Lemma 4.7)}
With high probability, algorithm \ref{alg:kmedian} preserves $(\epsilon_p, \delta_p)$ differential privacy and solves the k-medians problem with cost
\[O\left(M(1+\epsilon)\right)\OPT + O\left(\frac{Mk\Delta}{\epsilon_p}\ln n \left(\ln\left(\frac{e}{\delta_p}\right) + \frac{\ln(1/\epsilon)}{\ln(1+\epsilon)} \right) \right)\]
where the black-box k-medians algorithm used in the last step of algorithm \ref{alg:kmedian} has approximation factor $M$ and $\epsilon$ is a small, positive constant.
\end{lemma}

\begin{proof}
Combine the results of lemmas \ref{lem:approx}, \ref{lem:solutionapprox}, and \ref{lem:prelimcost}, we see that with high probability the optimal cost of the k-medians problem at the final line of algorithm \ref{alg:kmedian} is given by at most
\begin{equation*}
    \Big(\frac{2-\epsilon-2\epsilon^2}{1-\epsilon-\epsilon^2}\Big)  \OPT + \Big(  \frac{4\Delta k\ln(1/\epsilon)}{\epsilon_p}\left(\frac{\ln(n)}{\ln(1+\epsilon)} + 2\right) +  \frac{24 k \ln n \ln \left(\frac{e}{\delta_p}\right)}{\epsilon_p(1-\epsilon -\epsilon^2)}  +\frac{1-\epsilon}{1-\epsilon -\epsilon^2} \Big)\Delta.
\end{equation*}
For simplicity, we will use big-O notation going forward, so this is the same as
\[(2 + O(\epsilon))\OPT + O\left(\frac{k\Delta}{\epsilon_p}  \left(\ln (n)\frac{\ln(1/\epsilon)}{\ln(1+\epsilon)}  +  \ln n \ln\left(\frac{e}{\delta_p}\right)\right)\right).\]

Since the k-medians algorithm used at the last line has approximation factor $M$, the objective cost of the shifted k-medians problem is $M$ times that result. To obtain the objective cost for the original problem, we see that for any demand point $v \in D$, the distance between $v$ and a center is at most $d(c_v, v)$ and the distance from $c_v$ to a center, by triangle inequality. Therefore, the new objective cost only adds the shifting cost on top of the modified k-medians problem's objective cost, which only slightly affects the factor in front of \OPT:

\[(2M + 1 + (M+1)O(\epsilon))\OPT + O\left(\frac{Mk\Delta}{\epsilon_p}  \left(\ln (n)\frac{\ln(1/\epsilon)}{\ln(1+\epsilon)}  +  \ln n \ln\left(\frac{e}{\delta_p}\right)\right)\right)\]
which simplifies to the claim.
\end{proof}

\begin{lemma}
\label{lem:amendedtriangleineq_app}
\textbf{(Lemma 5.2)}
In any metric space and $p\ge 1$,
$(d(a,b))^p \le 2^{p-1}\left((d(a,c))^p + (d(b,c))^p\right)$.
\end{lemma}

\begin{proof}
Since the function $f(x) = x^p$ is convex over non-negative reals, and distance is a non-negative function,
\[\left(\frac{d(a,c)}{2} + \frac{d(b,c)}{2}\right)^p \le \frac{1}{2}\left(d(a,c)\right)^p + \frac{1}{2}\left(d(b,c)\right)^p. \]
Hence,
\begin{align*}
    \left(d(a,b)\right)^p &\le \left(d(a,c) + d(b,c)\right)^p\\
    &\le \frac{1}{2}\left(2d(a,c)\right)^p + \frac{1}{2}\left(2d(b,c)\right)^p\\
    &\le 2^{p-1}\left(\left(d(a,c)\right)^p + \left(d(a,c)\right)^p\right)
\end{align*}
The first inequality follows from the triangle inequality, and the second follows from the convexity claim above.
\end{proof}

\end{document}